\documentclass[%
 reprint,
 amsmath,amssymb,
 aps, prx,
]{revtex4-2}

\usepackage{graphicx}
\usepackage{dcolumn}
\usepackage{bm}
\usepackage{float}

\usepackage{algpseudocode}
\usepackage{algorithm}
\usepackage{physics}
\usepackage{color}
\usepackage{braket}
\usepackage[colorlinks=true, linkcolor=blue, citecolor=blue, urlcolor=blue]{hyperref}
\usepackage[capitalise]{cleveref}
\usepackage[all]{hypcap}

\usepackage{amsmath} 						
\usepackage{amssymb} 						
\usepackage{mathtools}			
\usepackage{mathrsfs}                       
\usepackage{amsthm}
\usepackage{amssymb}
\usepackage{hyperref}
\hypersetup{
    colorlinks = true,
    linkcolor = blue,
    allcolors=blue
    }
\usepackage{placeins}
\newtheorem{theorem}{Theorem}

\usepackage{algpseudocode}

\begin{document}
\global\setlength{\abovedisplayskip}{3pt}
\global\setlength{\belowdisplayskip}{3pt}
\global\setlength{\abovedisplayshortskip}{3pt}
\global\setlength{\belowdisplayshortskip}{3pt}
\setlength{\abovecaptionskip}{1pt}   
\setlength{\belowcaptionskip}{1pt}   
\thinmuskip=2mu
\preprint{APS/123-QED}

\title{On Minimizing Krylov Complexity Using Higher-Order Generators}


\author{Saud \v{C}indrak}
\email{saud.cindrak@tu-ilmenau.de}
\affiliation{%
Institute of Physics, Technische Universität
Ilmenau, Ilmenau, Germany
}%

\author{Kathy L{\"u}dge}
\affiliation{%
Institute of Physics, Technische Universität
Ilmenau, Ilmenau, Germany
}%

\date{\today}

\begin{abstract}
Krylov complexity provides a powerful framework for characterizing the dynamical evolution of quantum systems through the spreading of states in Krylov space. The motivation for this is rooted in the optimality of the Krylov basis for the analyzed cost function. In this work, we reinterpret the motivation for the Krylov basis from a dynamical perspective and show that it corresponds to a first-order approximation of the time-evolution operator. We extend this framework to higher-order generators and analytically disprove the optimality assumption by showing that an infinite-order generator can be constructed to exhibit smaller spread for arbitrary times. We propose a natural time scale for the construction of these higher-order generators and discuss results for matrices sampled from Gaussian Unitary Ensembles, demonstrating smaller Krylov complexity at all higher orders. These results extend the framework of Krylov complexity beyond the conventional Krylov basis by disproving the widely held assumption of optimality, extending the construction to higher-order generators, and introducing a physically motivated method for their construction. Our findings therefore suggest that previous statements and results concerning Krylov complexity may need to be reconsidered.
\end{abstract}

\maketitle

\textit{Introduction---}   
Krylov complexity \cite{PAR19} and its extension, Krylov spread complexity \cite{BAL22}, provide a framework for studying the dynamical evolution of quantum systems. Starting from an initial state $\ket{\psi_0}$ and a Hamiltonian $H$, one considers the Krylov space $\{H^n\ket{\psi_0}\}_{n=0}^{m-1}$, where $m$ denotes its grade. Orthonormalizing this set yields the Krylov basis, where states deeper in the space correspond to larger complexity. The resulting Krylov spread complexity characterizes how the state spreads under time evolution and exhibits characteristic dynamical signatures, such as a pronounced peak for chaotic systems before reaching a late-time average.
Due to these properties, Krylov operator and spread complexity have received considerable attention in recent years, with applications including quantum chaos \cite{PAR19, BAL25a, BAL22, ERD23}, Sachdev--Ye--Kitaev models \cite{BAL22, BHA24a, JIA21}, driven quantum systems \cite{NIZ23}, open quantum systems \cite{BHA22a, LIU23b}, unitary circuits \cite{CIN24a, SUC25}, time-dependent generators \cite{TAK25a}, and quantum machine computing \cite{DOM24, CIN25, CIN25a}, among many others \cite{AGU24, ALI23, ANE24, BAR19, CAM23, CAP24, CAO21, CHA23, CRA24, DYM21, FAN22, GAU24, GIL23a, GUO22A, HAS23, HE22, HEV22, HUH24, IIZ23, KIM22, LI24a, LIN22, MAG20, MUC22, NAN24, NIE06a, PAL23, PAT22, RAB23, VAS24, CIN24a, ZHO25b, SUC25, FAR26}. \cite{NAN25} is a review article that covers methods and applications.
The original work proposing Krylov spread complexity shows that the Krylov basis minimizes the associated cost function, providing both a physical and an information-theoretic motivation for its use \cite{BAL22}. As a result, nearly all previous studies employ this basis when computing Krylov spread complexity.

In this work, we propose an alternative motivation for the Krylov basis. We first show that Krylov spread complexity, also referred to as Krylov complexity or spread complexity, arises from a first-order approximation of the time evolution. This observation allows us to construct generators based on higher-order approximations of the time evolution, yielding a family of Krylov spaces parametrized by a time scale $\Delta t$. The infinite-order generator corresponds to the time-evolution operator up to time $\Delta t$  \cite{CIN24,CIN25, CIN25a}. In this way, the order interpolates between Hermitian first-order and unitary infinite-order generators. 
We then analytically disprove the assumption that the standard Krylov basis minimizes Krylov complexity \cite{BAL22}. In particular, we show that for any time an infinite-order generator can be constructed that yields a smaller Krylov complexity (see \cref{theorem_1}). These results are supported by numerical simulations of systems sampled from the Gaussian Unitary Ensemble (GUE).
Finally, we introduce a physically motivated choice of $\Delta t$ based solely on the spectral statistics of the quantum system and study its effect for a Gaussian unitary ensemble of dimension $N=50$. While increasing $\Delta t$ breaks tridiagonality, all higher-order generators consistently exhibit smaller Krylov complexity. Interestingly, their dynamics show strikingly similar behavior to the first-order generator, suggesting a broader framework for spread complexity.
With this work, we provide an interpretation of Krylov complexity from a dynamical perspective, where it emerges from a first-order generator, and show that the optimality of the Krylov basis does not generally hold, implying the need to reconsider its interpretation from an information-theoretic perspective. Finally, the introduction of higher-order generators naturally extends the framework of Krylov complexity, thereby opening new research directions in the study of operator dynamics.

\textit{Krylov Complexity---}
The time evolution in natural units ($\hbar =1$) of a quantum state $\ket{\psi(t)}$ satisfies the Schr\"odinger equation 
\begin{align}
\partial_t \ket{\psi(t)} = -iH\ket{\psi(t)}, \qquad \ket{\psi(0)} = \ket{\psi_0},
\label{eq:schroedinger}
\end{align}
where $H$ is the Hamiltonian generating the dynamics. The solution can be written as
\begin{align*}
\ket{\psi(t)} = e^{-iHt}\ket{\psi_0} = \sum_{k=0}^\infty \frac{(-iHt)^k}{k!}\ket{\psi_0},
\end{align*}
showing that time evolution lies in the span of $\{\ket{\psi_0}, H\ket{\psi_0}, H^2\ket{\psi_0},\ldots\}$. For a Hilbert space with finite dimension, there exists an $m\le\mathcal{H}$ such that
\begin{align}
\ket{\psi(t)} \in \mathrm{K}_m = \mathrm{Span}\{\ket{\psi_0},\,H\ket{\psi_0},\,\ldots,\,H^{m-1}\ket{\psi_0}\}.
\label{eq:K_m}
\end{align}
Orthonormalizing $\mathrm{K}_m$ yields the Krylov basis $\mathcal{B}_K=\{\ket{k_n}\}$ with $\langle{k_m}|{k_n}\rangle=\delta_{mn}$. For Hermitian $H$, the Lanczos algorithm \cite{LAN50} constructs the Krylov basis
through the three-term recurrence
\begin{align}
H\ket{k_n} = \mathrm{b}_{n+1}\ket{k_{n+1}} + \mathrm{a}_n\ket{k_n} + \mathrm{b}_n\ket{k_{n-1}},
\label{eq:H_kn}
\end{align}
where the Lanczos coefficients are determined by
\begin{align}
\mathrm{a}_n &= \bra{k_n}H\ket{k_n}, \nonumber \\
\mathrm{b}_{n+1} &= \left\|H\ket{k_n} - \mathrm{a}_n\ket{k_n} - \mathrm{b}_n\ket{k_{n-1}}\right\|.
\label{eq:a_b_n}
\end{align}
Using the orthonormality of the Krylov basis, the Lanczos recurrence relation (\cref{eq:H_kn}), and the Hermiticity of $H$, any Hamiltonian takes a tridiagonal form in $\mathcal{B}_K$ and can be rewritten as a tight-binding model \cite{NAN25}:
\begin{align*}
H = \sum_n \mathrm{a}_n\ket{k_n}\bra{k_n} + \mathrm{b}_{n+1}\left(\ket{k_{n+1}}\bra{k_n} + \ket{k_n}\bra{k_{n+1}}\right).
\end{align*}
Given an orthonormal basis (ONB) $\mathcal{B}=\{\ket{b_n}\}$, the Krylov spread complexity $\mathcal{C}(\mathcal{B}, t)$ is defined as
\begin{align}
\mathcal{C}(\mathcal{B},t)=\sum_n n\, \abs{\kappa_n(t)}^2 , \qquad \kappa_n(t)=\langle b_n|\psi(t)\rangle.
\label{eq:krylov_complexity}
\end{align}
The behavior of Krylov spread complexity is separated into early, medium, and late times, where early times see a fast increase and late times exhibit saturation, which is due to the state $\ket{\psi(t)}$ spreading across the full basis $\mathcal{B}$ \cite{PAR19, BAL22}. Therefore, one aims to minimize early- to medium-time complexity before the saturation time $t<T_{\mathrm{sat}}$. Balasubramanian \textit{et al.} argued that the Krylov basis minimizes early-time spread complexity, suggesting $\mathcal{B}_{\min} = \mathcal{B}_K$, thus further motivating the use of Krylov complexity from both a physical and an information-theoretic standpoint \cite{BAL22}.
For the Krylov basis $\mathcal{B}_K$, differentiating $i\kappa_n(t)=i\bra{k_n}{\psi(t)}\rangle$ with respect to time and using the Schr\"odinger equation (\cref{eq:schroedinger}), together with \cref{eq:H_kn} and the definition of $\kappa_n(t)$ in \cref{eq:krylov_complexity}, yields
\begin{align}
i\,\partial_t \kappa_n(t)
= b_n\,\kappa_{n-1}(t) + a_n\,\kappa_n(t) + b_{n+1}\,\kappa_{n+1}(t).
\label{eq:kappa_t}
\end{align}

\textit{First-order generator---}In the following, we motivate the use of the Krylov basis from the perspective of dynamical generators. For a short time step $\Delta t$, the evolution satisfies
\begin{align*}
\ket{\psi(\Delta t)} = e^{-iH\Delta t}\ket{\psi_0} \approx (I - iH\Delta t)\ket{\psi_0},
\end{align*}
which suggests the first-order generator
\begin{align*}
G^{(1, \Delta t)} &= I - iH\Delta t,\qquad \ket{g_0^{(1)}}=\ket{\psi_0}, \nonumber \\
\ket{g_{n+1}^{(1)}} &= G^{(1, \Delta t)}\ket{g_n^{(1)}}.
\end{align*}
For notational simplicity, we omit the $\Delta t$ dependence of the states $\ket{g_n^{(1)}}$.
When constructing an orthonormal basis from the sequence $\{\ket{g_n^{(1)}}\}$, the component proportional to $\ket{g_n^{(1)}}$ does not produce a new direction, since it lies entirely within the span of previously constructed vectors. The only part of $G^{(1, \Delta t)}\ket{g_n^{(1)}}$ that can generate a new direction is therefore the term proportional to $H\ket{g_n^{(1)}}$, giving
\(
\ket{g_{n+1}^{(1)}} \;\propto\; -i\Delta t\,H\ket{g_n^{(1)}}.
\)
After normalization, the scalar factor $\Delta t$ can be omitted, while the squared amplitudes used in the definition of Krylov complexity, $|\kappa_n(t)|^2$ (see \cref{eq:krylov_complexity}), are invariant under the global phase $-i$. Consequently, the effective generator reduces to
$
G^{(1, \Delta t)} = H.
$
This shows that the first-order generator exhibits the same Krylov spread complexity as the standard Krylov basis. This naturally motivates extending the framework of Krylov complexity to higher-order generators in the next step.

\textit{Higher-Order Generators---}
We assume that for any generator the first element is always the initial state, i.e.\ $\ket{g^{(p)}_0}=\ket{\psi_0}$, where $p$ indicates the order of the generator. A second-order approximation for small $\Delta t$ is
\begin{align*}
\ket{\psi(\Delta t)} \approx \left(I - i\Delta t\,H + \frac{(-i\Delta t)^2}{2}H^2\right)\ket{\psi_0},
\end{align*}
which leads to the generator
\begin{align*}
G^{(2, \Delta t)} &= I - i\Delta t\,H + \frac{(-i\Delta t)^2}{2}H^2, \nonumber \\
\ket{g_{n+1}^{(2)}} &= G^{(2, \Delta t)}\ket{g_n^{(2)}}.
\end{align*}
Applying the same argument as for the first-order generator, we remove the identity $I$ and one power of $-i\Delta t$, yielding $G^{(2,\Delta t)} = H - i\frac{\Delta t}{2}H^2$, which explicitly depends on $\Delta t$. After orthonormalization, the resulting space is represented as
$
\mathrm{K}_m^{(2,\Delta t)} = \mathrm{Span}\{\ket{k_0^{(2)}},\ldots,\ket{k_{m-1}^{(2)}}\}.
$ 
This generalizes to any order $p$ via
\begin{align*}
G^{(p,\Delta t)} &= \sum_{k=0}^{p} \frac{(i\Delta t\,H)^k}{k!}, \\
\ket{g_{n+1}^{(p)}} &= G^{(p,\Delta t)} \ket{g_n^{(p)}} .
\end{align*}
After orthonormalization we denote the basis by $\mathcal{B}_m^{(p,\Delta t)}$ and the corresponding space by $\mathrm{K}^{(p,\Delta t)}_m$, where all generators with $p>1$ depend on a time $\Delta t$. A particularly important case is the infinite-order generator
\begin{align*}
G^{(\infty, \Delta t)}&=e^{iH\Delta t}=U^{\Delta t}, \nonumber \\
\ket{g_{n+1}^{(\infty)}}&=U^{\Delta t}\ket{g_n^{(\infty)}},
\end{align*} 
recovering unitary time evolution, which mainly finds use in driven system \cite{NIZ23}, trotter circuits \cite{CIN24a, SUC25} and quantum machine learning\cite{CIN25, CIN25a}. Thus, this framework interpolates between the Hermitian first-order generator \cite{PAR19, BAL22} and the fully unitary infinite-order generator \cite{SUC25, CIN24a}. 
We can now define Krylov complexity constructed by the $p$th-order generator $G^{(p)}{(\Delta t)}$ as
\begin{align}
\mathcal{C}^{(p, \Delta t)}(t)=\sum_n n\, |\kappa^{(p)}_n(t)|^2 , \quad 
\kappa^{(p)}_n(t)=\langle k^{(p)}_n|\psi(t)\rangle .
\label{eq:krylov_complexity_extended}
\end{align}
The following theorem establishes one of the central results of this letter, namely that the Krylov basis does not minimize Krylov spread complexity. We prove this by contradiction, showing that the claim fails for a Krylov space of grade $m=3$ and then make a statement for any grade $m$.
\pagebreak
\begin{theorem}\label{theorem_1}
Let $H \in \mathbb{C}^{N\times N}$ be a Hamiltonian with Krylov grade $m=3$. Consider the orthonormal basis $\mathcal{B}_K = \{\ket{k_0}, \ket{k_1}, \ket{k_2}\}$ obtained by orthonormalizing the vectors $\{\ket{\psi_0},\, H\ket{\psi_0},\, H^2\ket{\psi_0}\}$, and fix a time $\tau$. 
The infinite-order generator induces the orthonormal basis $\mathcal{B}_U = \{\ket{g_0}, \ket{g_1}, \ket{g_2}\}$ obtained by orthonormalizing the vectors $\{\ket{\psi_0},\, U^{\Delta t}\ket{\psi_0},\, U^{2\Delta t}\ket{\psi_0}\}$, where $U^{\Delta t} = e^{-iH\Delta t}$. 
Then, for any $\tau$ there exists $\Delta t$ such that $\mathcal{C}^{(\infty,\Delta t)}(\tau) < \mathcal{C}^{(1,\Delta t)}(\tau).$
\end{theorem}

\begin{proof}
We begin by choosing $\Delta t=\tau$ for the construction of the generator and evaluate the spread at $\ket{\psi(\tau)}=\ket{\psi(\Delta t)}$. 
For both bases $\ket{k_0}=\ket{g_0}=\ket{\psi_0}$ holds. 
Expanding the time evolution yields
\begin{align*}
\ket{\psi(\tau)}
= \sum_{n=0}^2 \kappa_n^{(1)}(\tau)\ket{k_n}
= \sum_{n=0}^2 \kappa^{(\infty)}_n(\tau)\ket{g_n},
\end{align*}
with coefficients
$\kappa^{(1)}_n(\tau)=\braket{k_n|\psi(\tau)}$ 
and
$\kappa^{(\infty)}_n(\tau)=\braket{g_n|\psi(\tau)}.$
Since $\ket{k_0}=\ket{g_0}=\ket{\psi_0}$, it follows that 
$\kappa^{(1)}_0(\tau)=\kappa^{(\infty)}_0(\tau)$.
By construction of the basis $\mathcal{B}_U$, we use $\ket{\psi_0}$ and $U^{\tau}\ket{\psi_0}$ to obtain the first two basis elements, so at exactly this time the overlap with the third element vanishes, i.e.
$
|\kappa^{(\infty)}_2(\tau)|^2 = 0.
$
In contrast, $\ket{k_2}$ is fixed by the Hamiltonian and exists for all $t>0$, and we now show that the corresponding component $\kappa^{(1)}_2(\tau)$ is nonzero.
For $m=3$ the amplitudes according to \cref{eq:kappa_t} satisfy 
\begin{align}
i\,\dot{\kappa}^{(1)}_0(t) &= a_0\,\kappa^{(1)}_0(t) + b_1\,\kappa^{(1)}_1(t), \label{eq:phi0app}\\
i\,\dot{\kappa}^{(1)}_1(t) &= b_1\,\kappa^{(1)}_0(t) + a_1\,\kappa^{(1)}_1(t) + b_2\,\kappa^{(1)}_2(t), \label{eq:phi1app}\\
i\,\dot{\kappa}^{(1)}_2(t) &= b_2\,\kappa^{(1)}_1(t) + a_2\,\kappa^{(1)}_2(t). \label{eq:phi2app}
\end{align}
The initial conditions are $\kappa^{(1)}_0(0)=1$, $\kappa^{(1)}_1(0)=0$ and $\qquad \kappa^{(1)}_2(0)=0$. Because the Krylov grade is $3$, we have $b_1>0$ and $b_2>0$.
Evaluating Eq.~\eqref{eq:phi1app} at $t=0$ gives $i\,\dot{\kappa}^{(1)}_1(0)=b_1$, and therefore $\dot{\kappa}^{(1)}_1(0)=-i b_1\neq 0$ (\cref{eq:a_b_n}). This implies $\kappa^{(1)}_1(t)=-i b_1 t+\mathcal{O}(t^2)$, so $\kappa^{(1)}_1(t)\neq 0$ for sufficiently small $t=\tau>0$.
To discuss why $\kappa^{(1)}_2(\tau) \neq 0$, we consider the Taylor expansion around $t=0$.  
From Eq.~\eqref{eq:phi2app}, we have
$i\,\dot{\kappa}^{(1)}_2(0)=b_2 \kappa^{(1)}_1(0) + a_2 \kappa^{(1)}_2(0)=0$,  
and hence $\dot{\kappa}^{(1)}_2(0)=0$.   Next, we compute its second-order expansion around $t=0$.  
Differentiating Eq.~\eqref{eq:phi2app} gives
\(
i\,\ddot{\kappa}^{(1)}_2(t) = b_2\,\dot{\kappa}^{(1)}_1(t) + a_2\,\dot{\kappa}^{(1)}_2(t),
\)
and evaluating at $t=0$ yields
\(
i\,\ddot{\kappa}^{(1)}_2(0) = b_2(-i b_1),
\)
so that $\ddot{\kappa}_2(0) = -b_1 b_2 \neq 0$. Therefore,
\[
\kappa^{(1)}_2(t)
= -\frac{b_1 b_2}{2}\, t^2 + \mathcal{O}(t^3),
\]
which shows that $\kappa^{(1)}_2(t)\neq 0$ for all sufficiently small $t=\tau>0$.
Since $\kappa^{(\infty)}_2(\tau)=0$ and $\kappa^{(1)}_2(\tau)\neq 0$, while 
$|\kappa^{(1)}_0(\tau)| = |\kappa^{(\infty)}_0(\tau)|$, normalization implies
\begin{align*}
|\kappa^{(\infty)}_1(\tau)|^2 &= 1 - |\kappa^{(1)}_0(\tau)|^2, \nonumber\\
|\kappa^{(1)}_1(\tau)|^2 &= 1 - |\kappa^{(1)}_0(\tau)|^2 - |\kappa^{(1)}_2(\tau)|^2.
\end{align*}
Because $|\kappa^{(1)}_2(\tau)|^2>0$, we obtain
\(
|\kappa^{(1)}_1(\tau)|^2 < |\kappa^{(\infty)}_1(\tau)|^2.
\)
The corresponding Krylov complexities are
\begin{align*}
\mathcal{C}^{(1, \tau)}(\tau)
&= \sum_{n=0}^2 n\,|\kappa^{(1)}_n(\tau)|^2
= |\kappa^{(1)}_1(\tau)|^2 + 2|\kappa^{(1)}_2(\tau)|^2, \nonumber\\
\mathcal{C}^{(\infty, \tau)}(\tau)
&= \sum_{n=0}^2 n\,|\kappa^{(\infty)}_n(\tau)|^2
= |\kappa^{(\infty)}_1(\tau)|^2 .
\end{align*}
Using $|\kappa^{(\infty)}_1(\tau)|^2 = |\kappa^{(1)}_1(\tau)|^2 + |\kappa^{(1)}_2(\tau)|^2$, we obtain
\(
\mathcal{C}^{(1,\tau)}(\tau) > \mathcal{C}^{(\infty,\tau)}(\tau),
\)
which establishes the desired inequality by choosing $\Delta t=\tau$.
\end{proof}
This proof can be extended to any order $m$, where for every time $\tau =\Delta t$ one can define the infinite-order generator $\mathrm{G}^{(\infty, \Delta t)}$, which always exhibits $\kappa^{(\infty)}_{n>2}(\Delta t)=0$, whereas this does not hold for the first-order generator $\mathrm{G}^{(1)}$. This result therefore directly disproves the widely accepted assumption that the first-order generator minimizes spread complexity \cite{BAL22}.
\begin{figure}[t]
    \centering
    \includegraphics[scale = 1]{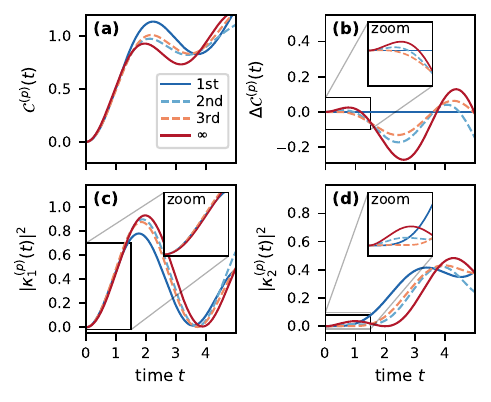}
    \caption{Spread complexity for a $3\times 3$ GUE Hamiltonian using higher-order Krylov generators. \textbf{a)} $\mathcal{C}^{(p, \Delta t)}(t)$ for orders $p=1,2,3,\infty$. \textbf{b)} Differences $\Delta\mathcal{C}^{(p, \Delta t)}(t)$ with a zoomed inset. Krylov amplitudes $|\kappa^{(p)}_{1}(t)|^{2}$ and $|\kappa^{(p)}_{2}(t)|^{2}$ in \textbf{c)} and \textbf{d)}.    }
    \label{fig:GUE3}
\end{figure}

What remains to be analyzed is the behavior in the interval $t\in[0,\Delta t]$. 
\Cref{fig:GUE3} shows results for a Hamiltonian $H\in\mathbb{C}^{3\times 3}$ sampled from a Gaussian unitary ensemble. We normalize such that $\sigma_{\max}(H)=1$ and choose $\Delta t=2$. In \cref{fig:GUE3}(a) we plot the Krylov spread complexity $\mathcal{C}^{(p)}(t)$ for the first-order (blue), second-order (light blue, dashed), third-order (orange, dashed), and infinite-order (red) generators. All complexities $\mathcal{C}^{(p)}(t)$ show an initial rise until $t\approx 1.6$, after which $\mathcal{C}^{(1)}(t)$ becomes the largest.
To highlight this behavior, \cref{fig:GUE3}(b) shows the difference $\Delta \mathcal{C}^{(p)}(t)=\mathcal{C}^{(p)}(t)-\mathcal{C}^{(p)}(t)$. Initially, higher-orders exhibit a slight increase above zero, indicating a small temporary enlargement of complexity. A zoom reveals that this increase disappears around $t\approx 1$, after which higher-order generators achieve smaller spread complexity, with reductions up to $\Delta\mathcal{C}\approx 0.25$. For ($t\gtrsim 3.8$) the higher order complexities again exceed the first-order complexity. To understand the early-time behavior, we examine the amplitudes $|\kappa_1(t)|^2$ and $|\kappa_2(t)|^2$ in \cref{fig:GUE3}(c–d). As previously argued $\kappa^{(p)}_0(t)$ is identical for all generators. We find that $|\kappa^{(p>1)}_1(t)|^2>|\kappa^{(1)}_1(t)|^2$ at intermediate times ($t\lesssim 3.8$) and smaller at late times ($t\gtrsim 3.8$), matching the two complexity regimes. Moreover, since the infinite-order generator enforces $\kappa_2^{(\infty)}(\Delta t)=0$, the state develops a larger transient weight in the second Krylov direction during $t\in[0,\Delta t)$, i.e., $|\kappa_2^{(\infty)}(t)|^2 > |\kappa_2^{(1)}(t)|^2$, as observed in \cref{fig:GUE3}(d). This explains the small initial complexity increase. For small $\Delta t$, this increase is negligible and of order $10^{-6}$. Therefore, we next discuss how to choose $\Delta t$ in a physically meaningful way to quantify the dynamical behavior.


\textit{Choosing the Time Step $\Delta t$---} 
A natural question is how $\Delta t$ should be selected to faithfully capture the dynamical behavior. 
To address this, we derive a characteristic timescale from the scrambling in the standard Krylov basis and scale $\Delta t$ accordingly.
To identify an appropriate value of $\Delta t$, we estimate when scrambling sets in by comparing consecutive Dyson terms $A_n(T)=(-iHT)^n/n!$ appearing in the expansion $\ket{\psi(T)}=\sum_{n=0}^{\infty}A_n(T)\ket{\psi_0}$.
Scrambling sets in when the $m$-th Dyson term—i.e., the first term outside the Krylov space—becomes comparable in norm to the $(m-1)$-st term, namely when
\begin{align}
{\|A_{m-1}(\tau_{\mathrm{scr}})\|} = {\|A_m(\tau_{\mathrm{scr}})\|}.
\label{eq:scrambling_condition}
\end{align}
Using the recursion relation
$
A_m(\tau_{\mathrm{scr}}) = A_{m-1}(\tau_{\mathrm{scr}})\,\frac{-iH\tau_{\mathrm{scr}}}{m}
$, and the fact that $H$ is Hermitian (and therefore normal), we have
\begin{align*}
\|A_m(\tau_{\mathrm{scr}})\|
&= \left\| A_{m-1}(\tau_{\mathrm{scr}})\,\frac{\tau_{\mathrm{scr}}}{m}H \right\| \nonumber \\
&\leq \|A_{m-1}(\tau_{\mathrm{scr}})\|\,\|H\|\frac{\tau_{\mathrm{scr}}}{m}.
\end{align*}
Inserting this expression into the scrambling condition \eqref{eq:scrambling_condition} yields
$
1 \leq \,\|H\|{\tau_{\mathrm{scr}}}/m.
$
Thus the scrambling time satisfies
$
\tau_{\mathrm{scr}} \geq {m}/{\|H\|}.
$
While this bound holds for any operator norm, in the special case of the spectral norm $\|A\| = \sigma_{\max}$, where $\sigma_{\max}$ is the largest singular value, we obtain
\begin{align}
\tau_{\mathrm{scr}} = \frac{m}{\|H\|}.
\label{eq:tau_scrambling}
\end{align}
Since the generator is applied $m$ times in approximating the time evolution, the scrambling time $\tau_{\mathrm{scr}}$ and the time step $\Delta t$ must satisfy
$
m\,\Delta t_{\mathrm{scr}} = \tau_{\mathrm{scr}},
$
which implies
$
\Delta t_{\mathrm{scr}} = \frac{\tau_{\mathrm{scr}}}{m} = \frac{1}{\|H\|}.
$
We index $\Delta t_{\mathrm{scr}}$ to emphasize that other choices are possible. For $\Delta t > \Delta t_{\mathrm{scr}}$, scrambling is observed in the Krylov basis for the standard approach, which depends only on the spectral statistics of the Hamiltonian.
To reflect a more realistic setting, we discuss in the following the influence of higher-order generators on Krylov complexity and tridiagonality for a larger system size as a function of the generator time scaled to the scrambling time, $\Delta t = \alpha \Delta t_{\mathrm{scr}}$.

\textit{Results---}
We begin by considering a random Hamiltonian $H\in\mathbb{C}^{N\times N}$ drawn from the Gaussian unitary ensemble with $N=50$, where the initial state $\ket{\psi_0}$ is a superposition of all eigenstates. For each realization, we construct higher-order Krylov generators of orders $p=1,2,3, \infty$ using the same construction and color coding as in \cref{fig:GUE3}, for a chosen time step $\Delta t$.
\begin{figure}[t]
    \centering
    \includegraphics[scale = 1]{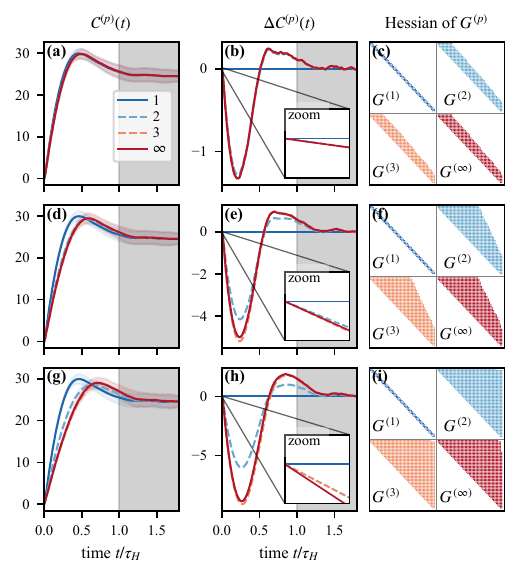}
    \caption{Spread complexity $\mathcal{C}^{(p)}(t)$ (first column), its deviation from first order $\Delta \mathcal{C}^{(p)}(t)$ (second column), and the Hessian (third column) for generators of order $p=1$ (blue), $p=2$ (dashed light blue), $p=3$ (dashed orange), and $p=\infty$ (red) for a $50\times50$ GUE Hamiltonian. Rows correspond to $\Delta t = 0.2\cdot\Delta t_{\mathrm{scr}}$ (a–c), $\Delta t = \Delta t_{\mathrm{scr}}$ (d–f), and $\Delta t = 1.5\cdot\Delta t_{\mathrm{scr}}$ (g–i). The grey region marks $t>\tau_H$, where $\mathcal{C}^{(p)}(t)$ approximately saturates.}
    \label{fig:GUE_50}
\end{figure}

Figure~\ref{fig:GUE_50} summarizes the behavior of higher-order Krylov generators for increasing time steps $\Delta t = 0.2\cdot\Delta t_{\mathrm{scr}}$ (a,b,c), $\Delta t = \Delta t_{\mathrm{scr}}$ (d,e,f), and $\Delta t = 1.5\cdot\Delta t_{\mathrm{scr}}$ (g,h,i).  
The time axis is normalized by the Heisenberg time $\tau_H = 2\pi/\overline{s}$, where $\overline{s}$ denotes the average level spacing of the ordered eigenvalues $H\ket{\phi_n}=\varepsilon_n\ket{\phi_n}$, given by $\overline{s}=\frac{1}{N}\sum_n(\varepsilon_{n+1}-\varepsilon_n)$.
The scrambling and Heisenberg times are $\tau_{\mathrm{scr}} = 50\Delta t = 0.037$ and $\tau_H = 5.836$, indicating that scrambling occurs well before saturation, consistent with the Heisenberg time scale. 
The first column (\cref{fig:GUE_50}(a,b,c)) shows a decrease in spread complexity when higher-order generators and larger values of $\Delta t$ are used. \cref{fig:GUE_50}(d,e,f) plots the deviations relative to the first-order generator, making this effect more transparent.  
For every $\Delta t$, we find that all higher-order Krylov complexities exhibit smaller spread complexity for early times $t/\tau_H \lesssim 0.6$. 
This is followed by a short region $0.6 \lesssim t/\tau_H \lesssim 1.3$ where they briefly exceed the first-order value, before all generators converge to the same late-time average for $t/\tau_H \gtrsim 1.2$. The grey region indicates $t>\tau_H$, which is often associated with the late-time regime.
Unlike the $3\times 3$ example in Fig.~\ref{fig:GUE3}, no initial increase in complexity is observed upon zooming in. This is due to the much larger ratio $\tau_H/\tau_{\mathrm{scr}}=157$, whereas in the three-dimensional case $\tau_H/\tau_{\mathrm{scr}}\approx 1$. However, such a definition is not meaningful for a matrix with only three eigenvalues, as it depends strongly on the random seed. The previous example was chosen to illustrate situations in which higher-order generators exhibit larger spread complexity, thereby guiding the motivation for physically grounded design of $\Delta t$.
\cref{fig:GUE_50}(c,f,i) shows the Hessians associated with each generator, displaying entries with magnitude $>10^{-6}$. 
The first-order generator is strictly tridiagonal and independent of $\Delta t$.
For $\Delta t=0.2\cdot \Delta t_{\mathrm{scr}}$, the higher-order Hessians deviate only slightly from this structure, with coefficients confined to a narrow band corresponding to short-range hopping in the Krylov basis, as illustrated by the comparison of the top-left panel (c) with the other panels in (c).
At $\Delta t=\Delta t_{\mathrm{scr}}$, the Hessians develop a broader triangular support, indicating that lower Krylov levels couple to several forward directions. For $\Delta t = 1.5\Delta t_{\mathrm{scr}}$ (last row in \cref{fig:GUE_50}), the Hessians become fully populated in the upper region, corresponding to an effective tight-binding model where each Krylov basis element couples to all forward directions, as also observed for Trotter circuits \cite{SUC25} and time-dependent generators \cite{TAK25a}.
 
\textit{Conclusion---}
This work shows that Krylov spread complexity can be viewed as a measure arising from the first-order generator of time evolution. We extend this perspective to time-dependent higher-order generators, where the order of the generator serves as a parameter interpolating between the first-order Hermitian generator \cite{BAL22} and the unitary infinite-order generator \cite{SUC25, CIN24a}.
For a system of Krylov grade $m=3$, we analytically disprove the widely held assumption that Krylov complexity is minimized by the Krylov space generated by the powers of the Hamiltonian $\{H^n\ket{\psi_0}\}_n$ for arbitrary time $\tau$ \cite{BAL22} and further support these findings through numerical simulations of Gaussian unitary ensembles. This result calls for reconsidering the information-theoretic motivation of Krylov spread complexity.  
We then derive a natural timescale by matching the scrambling of the first-order generator and use it to construct the higher-order generators, depending only on the spectral statistics of the Hamiltonian. Lastly, we analyze higher-order generators with respect to the scrambling time for larger system sizes, where our results show smaller complexity when higher-order generators are used and a loss of tridiagonality in the orthonormalization-induced tight-binding model.


\acknowledgments
We thank Zhuo-Yu Xian, Isaac Tesfaye, and Yukitoshi Motome for fruitful discussions and advice.

\appendix


\bibliography{lit}

\end{document}